\documentclass[11pt,a4paper]{article}
\usepackage{fullpage,amssymb,amsmath,amsfonts,amsthm,paralist,hyperref,graphicx,color}
\usepackage[normalem]{ulem}

\let\eps\varepsilon
\newtheorem{lemma}{Lemma}
\newtheorem{theorem}{Theorem}

\title{Finding the vertices of the convex hull, even unordered, takes $\Omega(n \log n)$ time---a proof by reduction from $\eps$-closeness}
\author{Herman Haverkort\thanks{Institut f\"ur Informatik, Universit\"at Bonn, Germany, cs.herman@haverkort.net. Written in response to a question in the computational geometry class of autumn 2018. I thank Rolf Klein for urging me to write it up.}}
\date{4 December 2018}

\begin{document}
\maketitle

\begin{abstract}
We consider the problem of computing, given a set $S$ of $n$ points in the plane, which points of $S$ are vertices of the convex hull of $S$. For certain variations of this problem, different proofs exist that the complexity of this problem in the algebraic decision tree model is $\Omega(n \log n)$. This paper provides a relatively simple proof by reduction from $\eps$-closeness.
\end{abstract}

\paragraph{The problem}
In the algebraic decision tree model of computation, it takes $\Omega(n \log n)$ steps in the worst case to compute the convex hull of a set of $n$ points in the plane. In text books~\cite{K,PS} this is proven by reduction from sorting as follows. Consider a set $A = a_1,...,a_n$ of $n$ distinct real numbers that need to be sorted. Now let $L = L_1,...,L_n$ be the set of points given by $L_i = (a_i, a^2_i)$. The convex hull of $L$, as a closed counterclockwise polygonal chain, contains the points of $L$ in order from left to right. Therefore, by computing the convex hull and then reading the $x$-coordinates of its vertices in order, we can obtain $A$ in sorted order. Hence, lower bounds for sorting also apply to the computation of convex hulls.

However, if we do not require the convex hull to be produced as a counterclockwise sequence of vertices (or edges), but settle for obtaining its vertices in arbitrary order, may we then be able to do so in $o(n \log n)$ time, that is, less than $\Omega(n \log n)$ time? The answer is no, as we will see now, using a reduction from the $\eps$-closeness problem.

\paragraph{The proof}
In its simplest form, the fixed-order algebraic decision tree model captures computer programs that take as input a sequence of real numbers, and consist of a number of instructions, each of which is of one of three types: 1. evaluate a polynomial function of the input variables, and use the sign of the result to choose which instruction to execute next; 2. output ``yes''; 3. output ``no''. The polynomial functions must have degree at most $d$, for some constant $d$: that is, each function is a sum of terms, where each term is a product of a constant and at most $d$ factors, where each factor is an input variable. 

The \emph{$\eps$-closeness} problem is the following: given an unordered multiset $A$ of real numbers $a_1,...,a_n$ and a positive real number $\eps$, decide whether $A$ contains two numbers that differ less than $\eps$. In other words: is there a pair of indices $i \neq j$ such that $0 \leq a_j - a_i < \eps$, or equivalently, $a_i \leq a_j < a_i + \eps$? Any fixed-order algebraic decision tree algorithm that solves the $\eps$-closeness problem takes $\Omega(n \log n)$ time in the worst case~\cite{PS}.

Let us define the \emph{any-point-inside} problem as follows: given an unordered multiset $S$ of $n$ points in the plane (given as $2n$ real coordinates), decide whether any point of $S$ lies in the interior of the convex hull of $S$. Let $api(n)$ be the number of steps it takes in the worst case to solve this problem in the fixed-order algebraic decision tree model.

\begin{lemma}\label{lem:allowingequals}
Given a multiset of real numbers $A = a_1,...,a_n$ and a real number $\eps$, we can decide in $api(2n)$ time whether there are two indices $i \neq j$ such that $0 < a_j - a_i < \eps$, or equivalently, $a_i < a_j < a_i + \eps$.
\end{lemma}
\begin{proof}
Let $S(A, \eps)$ be $L \cup T$, where $L = L_1,...,L_n$ is the set of points given by $L_i = (a_i, a^2_i)$, and $T = T_1,...,T_n$ is the set of points given by $T_i = (a_i + \eps/2, (a_i + \eps/2)^2 + \eps^2/4)$. We will see shortly that $S(A, \eps)$ is constructed such that all of its points will appear on the boundary of its convex hull, unless some point $T_i$ is ``hidden'' by a nearby point $L_j$ where $0 < a_j - a_i < \eps$. Thus, to determine whether there are two indices $i \neq j$ such that $0 < a_j - a_i < \eps$, we simply solve the any-point-inside problem on $S(A, \eps)$ in $api(2n)$ time.

To prove the correctness of this reduction, we start with the following observations. For $i \in \{1,...,n\}$, let $R_i$ be the point $(a_i + \eps, (a_i + \eps)^2)$. Note that $L_i$ and $R_i$ lie on the parabola $y = x^2$, which we call the \emph{outer} parabola. The point $T_i$ lies exactly half-way between $L_i$ and $R_i$ on the parabola $y = x^2 + \eps^2/4$, which we call the \emph{inner} parabola. Moreover, the slope of the inner parabola at $T_i$ is $2 a_i + \eps$, which is exactly the slope of the line $\ell_i$ through $L_i$, $T_i$ and $R_i$, so $\ell_i$ is a tangent to the inner parabola. 

Now I claim the following: (i) if there are no $i$ and $j$ such that $a_i < a_j < a_i + \eps$, then all $2n$ points of $S(A, \eps)$ are on the boundary of the convex hull of $S(A, \eps)$; (ii) if there are $i$ and $j$ such that $a_i < a_j < a_i + \eps$, then $S(A, \eps)$ contains at least one point, namely $T_i$, that is not on the boundary of the convex hull of $S(A, \eps)$. Both claims are illustrated in Figure~\ref{fig:fig}.

\begin{figure}
\centering\includegraphics[width=\hsize]{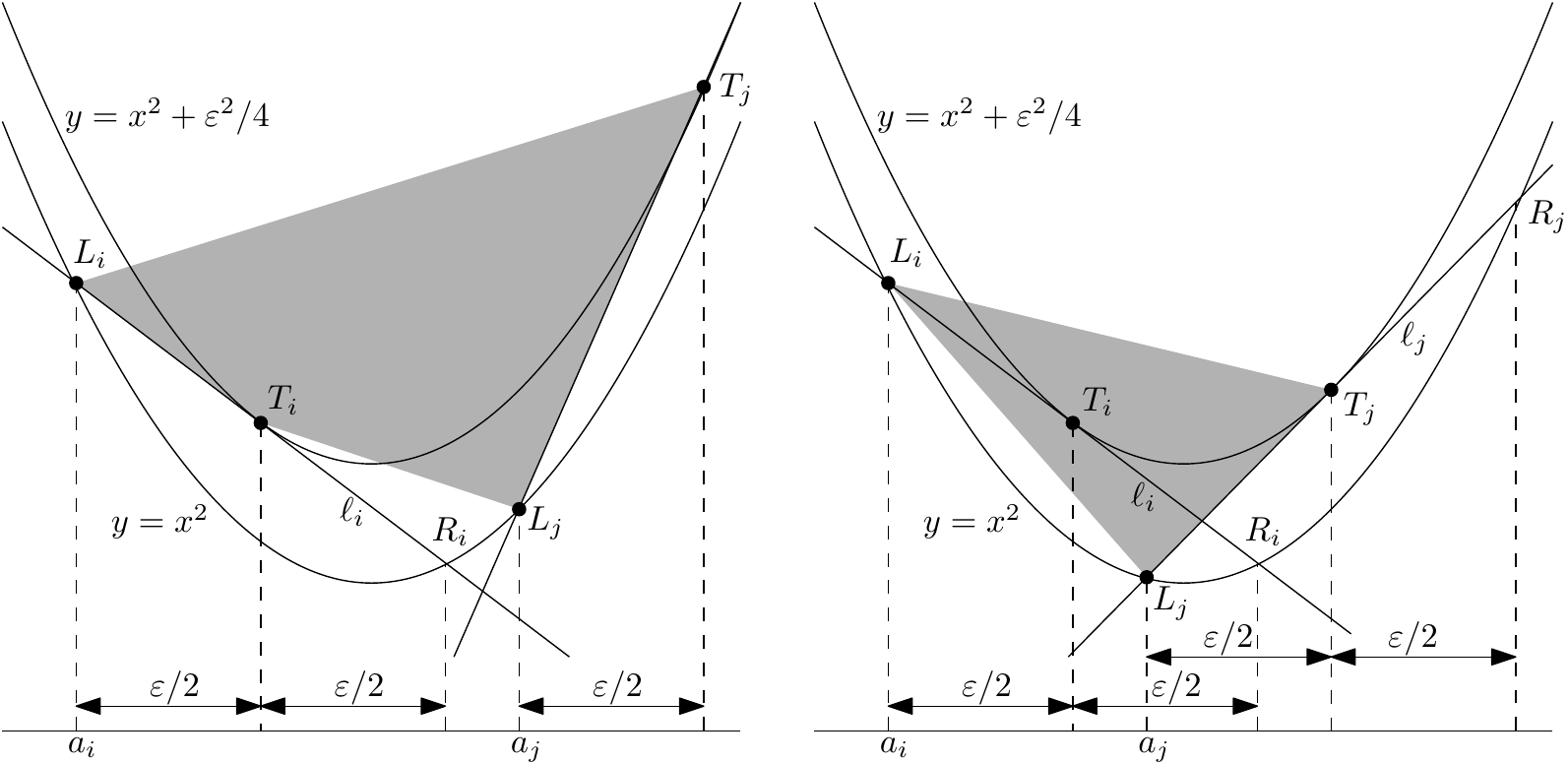}
\caption{Left: if all numbers in $A$ are at least $\eps$ apart, then all points of $L$ and $T$ are on the boundary of the convex hull. Right: if $a_i < a_j < a_i + \eps$, then $T_i$ lies in the interior of the triangle $L_i L_j T_j$.}
\label{fig:fig}
\end{figure}

Proof of (i): assume there are no $i$ and $j$ such that $a_i < a_j < a_i + \eps$. We will show that both $L_i$ and $T_i$ lie on the boundary of the convex hull for each $i$. Recall that the line $\ell_i$ through $L_i$, $T_i$ and $R_i$ is a lower tangent to the inner parabola and intersects the outer parabola in $L_i$ and $R_i$. So, if $S(A, \eps)$ would contain any points that lie strictly below $\ell_i$, these points would have to be points $L_j$ on the outer parabola strictly between $L_i$ and $R_i$. However, looking at the $x$-coordinates of these points, we would then find $a_i < a_j < a_i + \eps$, contradicting our assumptions. It follows that both $L_i$ and $T_i$ lie on the boundary of the convex hull, as there is a closed half-plane with $L_i$ and $T_i$ on the boundary and all other points of $S(A, \eps)$ inside.

Proof of (ii): we will show that $T_i$ lies in the interior of the triangle $L_i L_j T_j$. First, because $a_i < a_j < a_i + \eps$, the point $L_j$ lies on the outer parabola strictly between $L_i$ and $R_i$, so $L_j$ lies strictly under the line $\ell_i$ through $L_i$, $T_i$, and $R_i$, and vice versa, $T_i$ lies strictly above the line through $L_i$ and $L_j$. Second, recall that the line through $L_j$ and $T_j$ is a lower tangent to the inner parabola, touching it in $T_j$, so $T_i$ lies strictly above it. Third, the line through $L_i$ and $T_i$ is a lower tangent to the inner parabola, so $T_j$ lies strictly above it, and vice versa, $T_i$ lies strictly below the line through $L_i$ and $T_j$. It follows that $T_i$ lies in the interior of the triangle $L_i L_j T_j$, and thus, not on the boundary of the convex hull of $S(A, \eps)$.

Thus, $S(A, \eps)$ contains a point that lies in the interior of the convex hull of $S(A, \eps)$ if and only if there are $i$ and $j$ such that $a_i < a_j < a_i + \eps$. This proves the correctness of the reduction and thus proves the lemma.
\end{proof}

\begin{lemma}\label{lem:catching-identical-points}
Given a multiset of real numbers $A = a_1,...,a_n$ and a real number $\eps$, we can decide $\eps$-closeness in $api(2n)$ time, that is, we can decide in $api(2n)$ time whether there are two indices $i \neq j$ such that $0 \leq a_j - a_i < \eps$, or equivalently, $a_i \leq a_j < a_i + \eps$.
\end{lemma}
\begin{proof}
We first decide, in $api(2n)$ time by Lemma~\ref{lem:allowingequals}, whether there are two indices $i \neq j$ such that $0 < a_j - a_i < \eps$. If this is the case, we answer ``yes''. Otherwise, we have now established that each pair of numbers in $A$ differs by either zero, or at least $\eps$, and what is left to decide is whether there is any pair with difference zero. Let $A' = a'_1,...,a'_n$ be given by $a'_i = a_i + \eps i/(2n)$. Thus, $|a'_j - a'_i|$ and $|a_j - a_i|$ differ by less than $\eps/2$ for any $i$ and $j$. Observe that, given that there are no $i, j$ such that $0 < a_j - a_i < \eps$, we now have $i \neq j$ and $a_j - a_i = 0$ if and only if $0 < |a'_j - a'_i| < \eps/2$. By Lemma~\ref{lem:allowingequals}, we can decide whether there is any pair $i, j$ for which the latter is the case in $api(2n)$ time. 
\end{proof}

Since we have an $\Omega(n \log n)$-time lower bound for $\eps$-closeness, we must now conclude $api(2n) = \Omega(n \log n)$, so $api(n) = \Omega(\frac n2 \log \frac n2) = \Omega(n \log n)$, and we obtain the following theorem:

\begin{theorem}\label{thm:boundary}
Given a multiset $S$ of $n$ points in the plane, any fixed-order algebraic decision tree algorithm takes $\Omega(n \log n)$ steps in the worst case to decide whether $S$ contains any point that lies in the interior of the convex hull of $S$.
\end{theorem}

If, in the $\eps$-closeness problem, we replace $<$ by $\leq$, then an $\Omega(n \log n)$-time lower bound can be constructed in the same way as for the original $\eps$-closeness problem. Furthermore, we can adapt the proof of Lemma \ref{lem:allowingequals} to prove (i) if there are no $i$ and $j$ such that $a_i < a_j \leq a_i + \eps$, then all $2n$ points of $S(A, \eps)$ are extreme points (vertices) of the convex hull of $S(A, \eps)$; (ii) if there are $i$ and $j$ such that $a_i < a_j \leq a_i + \eps$, then $S(A, \eps)$ contains at least one point, namely $T_i$, that is not an extreme point of the convex hull of $S(A, \eps)$. Thus, we also get the $\Omega(n \log n)$-time lower bound for the following problem:

\begin{theorem}\label{thm:convexposition}
Given a multiset $S$ of $n$ points in the plane, any fixed-order algebraic decision tree algorithm takes $\Omega(n \log n)$ steps in the worst case to decide whether $S$ is in convex position (that is, whether $S$ is exactly the set of vertices of the convex hull of $S$).
\end{theorem}

To output the vertices of the convex hull of $S$, we would need to extend the model of computation so that it admits other output than just ``yes'' or ``no''. No matter how exactly we do that, it is clear that we cannot output the vertices of the convex hull in less than $\Omega(n \log n)$ time, if we cannot even decide whether or not all of $S$ should be output in less than $\Omega(n \log n)$ time.

\subsubsection*{Other proofs}

\enlargethispage{\baselineskip}
\textbf{Preparata and Shamos} provide another proof that it takes $\Omega(n \log n)$ steps to decide whether a set of points is in convex position (Theorem 3.3 in Section 3.2 in their book~\cite{PS}). Their proof is based directly on the lower bounds for membership tests for a set $W$ in high-dimensional space with $n!$ connected components---the same lower bounds that also underly the lower bounds for $\eps$-closeness that we used above. The core of their proof is therefore an analysis of the number of connected components of $W$, the subset of possible inputs that corresponds to sets of points in convex position (\cite{PS}, p102--103). The key argument considers pairs of inputs in convex position that must lie in different components, because they cannot be continuously transformed into each other without passing through a configuration in which three points are collinear\footnote{The reader who wants to verify the details of the proof should beware of minor typing mistakes that make it look as if permutations on $N$ integers are applied to integers from 0 to $2N-1$ or from 0 to $N^2-1$.}. Note that a triple of collinear points immediately rules out convex position, but it is does not imply that one of the points lies in the interior of the convex hull (they could all lie on the edge). Thus, the proof by Preparata and Shamos is a bit more specific than ours: they prove our Theorem~\ref{thm:convexposition}, but not our Theorem~\ref{thm:boundary}.

\textbf{Kirkpatrick and Seidel}~\cite{KS} prove the following (Theorem 5.3): given a set $S$ of $n$ distinct points and a natural number $h \leq n$, any fixed-order algebraic decision tree algorithm requires, in the worst case, $\Omega(n \log h)$ steps to verify that the convex hull of $S$ has $h$ vertices. When $h$ is polynomial in $n$, this bound is equivalent to $\Omega(n \log n)$. Clearly, if we could compute in $o(n \log n)$ time which points of $S$ are vertices of the convex hull, then we could also count them in $o(n \log n)$ time. Thus, the theorem by Kirkpatrick and Seidel implies an $\Omega(n \log n)$ lower bound for the problem of computing the vertices of the convex hull. 

Like the proof of our Lemma~\ref{lem:catching-identical-points}, the proof by Kirkpatrick and Seidel uses a reduction from a problem on a multiset $A$ of real numbers, namely the problem of verifying the number of distinct numbers in the set (in the proof of Lemma~\ref{lem:catching-identical-points}, the problem is to verify that \emph{all} numbers of the set are distinct). To this end, each number $a_i$ in $A$ is mapped to a point on a parabola and perturbed by moving it over a distance that grows with $i$. Thus, whenever there are $k > 1$ points that represent the same number, these points all become distinct; moreover, the perturbations are carefully chosen such that $k-1$ of these points are no longer a vertex of the convex hull. The perturbations could be implemented by adapting all nodes that evaluate the coordinates of the input points in the decision tree, so that they evaluate the perturbed coordinates instead of the original coordinates. The perturbations must be small enough, so that points that were distinct already have no effect on each other---in our proof of Lemma~\ref{lem:catching-identical-points}, we ensure this by limiting the perturbation in the horizontal direction to $\eps/2$.

Of course, smaller perturbations would also work and lead to the same end result. This is where Kirkpatrick and Seidel use a clever trick: rather than using perturbations of a fixed size, they describe how to adapt the decision tree (making it only slightly higher) so that it effectively takes the same decisions that it would take with \emph{any} small enough perturbations. Thus, Kirkpatrick and Seidel can realize a reduction from the multiset size verification problem, where no prespecified difference threshold $\eps$ can be used to determine the size of the perturbations.

\end{document}